\documentclass[reqno,a4paper,11pt,centertags]{amsart}
\usepackage{amsmath,amsthm,amscd,amssymb,latexsym,upref,stmaryrd,cite}
\usepackage{geometry,subfigure}
\usepackage{graphicx}
\usepackage{hyperref}
\newtheorem{theorem}{Theorem}[section]
\newtheorem{lemma}[theorem]{Lemma}

\newtheorem{corollary}[theorem]{Corollary}

\newtheorem{ip}{Inverse Problem}[section]

\newtheorem{remark}{Remark}[section]

\title{inverse spectral problem for the Schr\"odinger operator on the square lattice}
\author{Dongjie Wu\textsuperscript{1}}
\author{Chuan-Fu Yang\textsuperscript{2}}
\author{Natalia Pavlovna Bondarenko\textsuperscript{3}}

\begin{document}
	\maketitle
	{
		\footnotetext[1]{Department of Applied Mathematics, School of Mathematics and Statistics, Nanjing University of Science and Technology, Nanjing, 210094, Jiangsu, China, Email: wudongjie@njust.edu.cn}
		\footnotetext[2]{Department of Applied Mathematics, School of Mathematics and Statistics, Nanjing University of Science and Technology, Nanjing, 210094, Jiangsu, China, Email: chuanfuyang@njust.edu.cn}
		\footnotetext[3]{S.M. Nikolskii Mathematical Institute, Peoples' Friendship University of Russia (RUDN University), 6 Miklukho-Maklaya Street, Moscow, 117198, Russian Federation, Email: bondarenkonp@sgu.ru}
	}
	{\noindent\small{\bf Abstract:}
		We consider an inverse spectral problem on a quantum graph associated with the square lattice. Assuming that the potentials on the edges are compactly supported and symmetric, we show that the Dirichlet-to-Neumann map for a boundary value problem on a finite part of the graph uniquely determines the potentials. 
		We obtain a reconstruction procedure, which is based on the reduction of the differential Schr\"odinger operator to a discrete one. As a corollary of the main results, it is proved that 
		the S-matrix for all energies
		in any given open set in the continuous spectrum uniquely specifies the potentials on the square lattice.
	}
	\vspace{1ex}
	
	
	\section{introduction}\label{sec:introduction}
	Recently, there have been a lot of studies on \textit{quantum graphs}, which are one-dimensional Schr\"odinger (Sturm-Liouville) operators $-\frac{d^2}{dz^2} + q_e(z)$ acting on the edges of a metric graph, while some matching conditions are imposed at the graph vertices. Such operators are used for modeling various processes on graph-like structures in physics, mechanics, chemistry, and other applications. Expositions of spectral theory results for quantum graphs can be found, e.g., in the monographs \cite{Cve,BCFK06, Ber,Chu} and references therein.
	
	This paper is mostly focused on \textit{inverse} spectral and scattering problems. Such problems consist in the reconstruction of unknown operator characteristics from spectral information. Till now, inverse problems have been studied for several types of quantum graphs. Reconstruction of differential operators on \textit{compact} graphs has been investigated in \cite{AK23, Bel, Bon, Gut, Kur, Kur10, Yur1, Yur3, Yur2} and other studies. Inverse spectral-scattering problems on \textit{non-compact} graphs with finite and infinite edges were solved, e.g., in \cite{But, Ign, MT12}. 
	
	In recent years, Schr\"odinger operators on infinite \textit{periodic} graphs have attracted considerable attention of scholars in connection with applications in material studies and nanotechnology (see \cite{Kuc, Kor1} and references therein). Spectral properties of such operators were studied in \cite{Kuc, Kor1, Kor2, Luo1, Luo2} and other papers. 
	We also mention that in \cite{Exn} some spectral theory issues were considered for infinite quantum graphs of general structure (not necessarily periodic).	
	However, to the best of the authors' knowledge, there were no studies on inverse problems for differential Schr\"odinger operators on periodic graphs except for the two preprints \cite{And4, And5} of Ando et al. Let us discuss their results in more detail.
	
	The two preprints \cite{And4, And5} present the same results on the inverse scattering for the Schr\"odinger operator on the hexagonal lattice with finitely supported potential. The authors of \cite{And4, And5} have shown that, in this case, the scattering matrix (S-matrix) uniquely determines the Dirichlet-to-Neumann (D-N) map for a boundary value problem on a finite part of the graph. Furthermore, the differential (continuous) operator was reduced to a difference (discrete) one and the potentials on the graph edges were reconstructed from the D-N map. Thus, it has been shown that the S-matrix uniquely specifies a finitely supported potential on the hexagonal lattice. 	
	The preprints \cite{And4, And5} continue the previous studies of their authors \cite{And1, And2, And3, Iso1, Iso2}, in which analogous ideas and methods were developed for the inverse scattering on \textit{discrete} periodic graphs. The difference between \cite{And4} and \cite{And5} is that, in \cite{And5}, the preliminary steps up to determining the D-N map by the S-matrix are implemented for the hexagonal lattice, while in \cite{And4} they are provided for a general periodic lattice. This opens a perspective of studying inverse spectral problems for different types of periodic graphs.
	
	This paper is concerned with a family of one-dimensional Schr\"{o}dinger
	operators $ -\frac{d^{2}}{dz^{2}}+q_{e}(z)$ defined on the edges of the square lattice as in Figure~\ref{1}, assuming
	the Kirchhoff conditions at the vertices (the details are given in Section~\ref{sec:edge}). Here, $z$ varies over the interval $(0, 1)$ and $e\in E$, where $E$ is the set of all edges of the square lattice. 
	
	Let us impose the following assumptions on the potentials.
	
	\smallskip
	
	\textbf{(Q-1)}  $q_{e}(z)$ is real-valued, and $q_{e} \in L^{2}(0,1)$.
	
	\textbf{(Q-2)}  $q_{e}(z) = 0$ on $(0, 1)$ except for a finite number of edges.
	
	\textbf{(Q-3)}  $q_{e}(z) = q_{e}(1-z)$ for $z \in (0, 1)$.
	
	\begin{figure}[h]
		\centering
		\includegraphics[width=7.5cm,height=4cm]{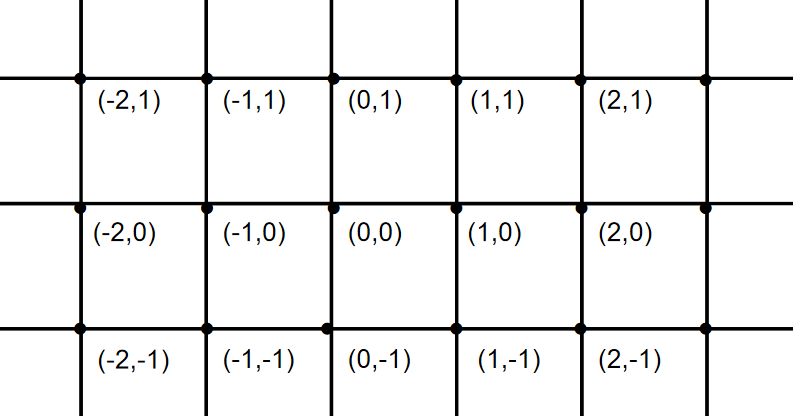}
		\caption{}
		\label{1}
	\end{figure}

	Since the support of the potential $q$ is finite, we can choose a sufficiently large square domain $D_N := \{ n_1 + \mathrm{i} n_2 \colon 0 \le n_1, n_2 \le N\}$ such that $\mbox{supp}\, q$ is located inside of $D_N$ and consider the edge Dirichlet-to-Neumann map $\Lambda_E$ associated with this domain (see Section~\ref{sec:edge} for details). This paper is devoted to the following inverse spectral problem.	

	\begin{ip} \label{ip:DNe}
		Given the edge D-N map $\Lambda_E$, find the potential $q$. 
	\end{ip}

	The main results of this paper are the uniqueness theorem (Theorem~\ref{thm:edgeDN}) for Inverse Problem~\ref{ip:DNe} and the reconstruction procedure in Section~\ref{sec:construct}. For solving Inverse Problem~\ref{ip:DNe}, we use the relation between the edge D-N map and the vertex D-N map, which is defined in Section~\ref{sec:vertex}. In other words, we reduce the continuous problem to the discrete one. Consequently, we develop a constructive algorithm for the recovery of the potentials $q_e$ from the vertex D-N map. Note that, although we follow the general idea of Ando et al, our reconstruction procedure is different from \cite{And4, And5}. First, our algorithm is based on specific properties of the square lattice. Second, the authors of \cite{And4, And5} apply the strategy of recovering the potentials along any zigzag line, which perfectly works for discrete operators with constant coefficients \cite{And3, Iso2} but causes difficulties for the reduced vertex Laplacian whose coefficients depend on the unknown potentials (see Remark~\ref{rem:compare} for details). In this paper, we step-by-step recover $q_e$ together with the Laplacian coefficients required at the next steps.
	For reconstruction of the potential on each fixed edge, we use the classical results of the inverse spectral theory for the Schr\"odinger operators on finite intervals (see, e.g., \cite{Fre}).
	Finally, our results imply the uniqueness for solution of the inverse scattering problem by the S-matrix on the square lattice.
	
	It is worth mentioning that Inverse Problem~\ref{ip:DNe} can be treated as the inverse spectral problem on a finite metric graph by the Weyl matrix associated to the boundary vertices. However, our results are novel in this direction. It is well-known the the Weyl matrix uniquely specifies the Schr\"odinger operator on a tree-graph (see \cite{Bel, Yur1}). But, for graphs with cycles, this is not the case in general. Even for the simplest graphs with loops additional data are required (see \cite{Kur10, Yur2}). Our paper provides a new class of finite quantum graphs whose potentials are uniquely determined by the Weyl matrix. Here, the symmetry of the potentials (Q-3) is crucial.		
	
	The paper is organized as follows. In Section \ref{sec:edge}, we give the definitions related to the edge Schr\"odinger operator and the edge D-N map. In Section \ref{sec:vertex}, we define the reduced vertex Laplacian and the vertex D-N map, and study the relation between the edge D-N map and the vertex D-N map. In Section \ref{sec:solution}, some auxiliary solutions of the vertex Schr\"{o}dinger equation are obtained. In Section \ref{sec:construct}, we reconstruct the potentials from the D-N map. In Section \ref{sec:scattering}, the inverse scattering by the S-matrix is discussed.

	\section{Edge Laplacian and Dirichlet-to-Neumann map} \label{sec:edge}
	
	Let us define an infinite square lattice with the vertex set
	$$
		V := \{ n_1 + \mathrm{i} n_2 \colon n_1, n_2 \in \mathbb Z \} 
	$$
	and the edge set 
	$$
		E := \{ (v, v + 1), (v, v + \mathrm{i}) \colon v \in V\}.
	$$
	In other words, we consider the vertices as points on the complex plane $\mathbb C$ and suppose that two vertices are joined by an edge if the distance between them equals $1$. For two vertices $w, v \in V$, the notation $w \sim v$ means that there exists an edge $e \in E$ such that $v, w$ are end points of $e$.
	
	Let each edge $e$ be endowed with arclength metric and identified with the interval $(0, 1): e = \{(1 - z)e(0) + ze(1) \colon 0 \leqslant z \leqslant 1\}$, where $e(0), e(1) \in V$.
	Put 
	$$
	E_{v} := E_{v}(0) \cup E_{v}(1), \quad  E_{v}(i) := \{e \in E\colon e(i) = v\},\  i = 0, 1.
	$$
	
	Let $\mathcal E \subseteq E$.
	Then, we call $f = \{ f_e \}_{e \in \mathcal E}$ \textit{a function} on $\mathcal E$ if each $f_e$ ($e \in \mathcal E$) is a function on $[0,1]$. We will write that $f = \{ f_e \}_{e \in \mathcal E} \in \mathcal A(E)$ if $f_e \in \mathcal A[0,1]$ for all $e \in \mathcal E$, where $\mathcal A$ is any functional class, e.g., $\mathcal A = C, L^2$, etc. For $v \in V$, $e \in E_v \cap \mathcal E$, and $f = \{ f_e\}_{e \in \mathcal E}$, introduce the notations
	$$
		f(v,e) = \begin{cases}
		f_e(0), & e(0) = v, \\
		f_e(1), & e(1) = v,
		\end{cases}, \qquad
		\partial_e f(v) = \begin{cases}
		-f_e'(0), & e(0) = v, \\
		f_e'(1), & e(1) = v.
		\end{cases}
	$$
	In other words, $f(v,e)$ is the value of the function $f$ and $\partial_e f(v)$, of its derivative in the vertex $v$ along the edge $e$.
	
	Let $v \in V$ be a vertex such that $E_v \subseteq \mathcal E$. Then, we say that $f = \{ f_e \}_{e \in \mathcal E}$ satisfies \textit{the Kirchhoff conditions} at $v$ if
	
	\smallskip
	
	\textbf{(K-1)} $f$ is continuous at $v$, that is, $f(v,e_1) = f(v,e_2)$ for any $e_1, e_2 \in E_v$. In this case, we denote $f(v) = f(v,e)$, $e \in E_v$.

	\textbf{(K-2)} $f\in C^1(E_v)$ and $\sum_{e\in E_{v}} \partial_e f(v)=0$.
	
	\smallskip
			
	We will also use the notation $f(v) = f(v,e)$ if $e$ is the only edge in $\mathcal E$ incident to the vertex $v$, that is, $\mathcal E \cap E_v = \{ e \}$.
	
	Let us define the Schr\"odinger operator on a finite subgraph of the square lattice. Suppose that the potential $q = \{ q_e\}_{e \in E}$ satisfies the assumptions (Q-1), (Q-2), and (Q-3).
	Let $\mathring{\Omega}$ be the vertex set of some connected finite subgraph of the lattice $(V, E)$. Denote 
\begin{gather} \label{defOmega}
\partial \Omega :=\{ v\notin \mathring{\Omega}\colon \exists w\in \mathring{\Omega}\colon w\sim v \}, \quad \Omega := \mathring{\Omega} \cup \partial \Omega, \\ \nonumber E_{\Omega} := \{ e = (v,w) \in E \colon v, w \in \Omega \: \text{and} \: (v \in \mathring{\Omega} \: \text{or} \: w \in \mathring{\Omega}) \}, \\ \nonumber
\mathring{E}_{\Omega} := \{ e = (v,w) \in E_{\Omega} \colon v, w \in \mathring{\Omega}\}, \quad \partial E_{\Omega} := E_{\Omega} \setminus \mathring{E}_{\Omega}.
\end{gather}

Thus, $E_{\Omega}$ is the edge set of the subgraph with the vertices $\Omega$.
The notations $\mathring{\Omega}$, $\partial \Omega$, $\mathring{E}_{\Omega}$, and $\partial E_{\Omega}$ are used for the interior vertices, the boundary vertices, the interior edges, and the boundary edges, respectively, of the subgraph $(\Omega, E_{\Omega})$.

Denote by $\Delta_E$ and call \textit{the edge Laplacian} the second derivative operation along each edge:
$$
\Delta_E u_e(z) = \frac{d^2}{dz^2} u_e(z), \quad e \in E.
$$

For the subgraph $(\Omega, E_{\Omega})$, consider the Dirichlet boundary value problem for the edge Laplacian
\begin{equation}
\begin{cases}
(-\Delta_E+q_{e}-\lambda)u_e=0,\ in\  \mathring{E}_{\Omega},\\
u=f,\ on \ \partial \Omega, 
\end{cases}\label{BVP1}
\end{equation}
where a function $u = \{ u_e\}_{e \in E_{\Omega}} \in H^2(E_{\Omega})$ is assumed to fulfill the Kirchhoff conditions at every interior vertex $v \in \mathring{\Omega}$ and $f = \{ f(v)\}_{v \in \partial \Omega}$ is a function given on the set of the boundary vertices.

Denote by $(-\Delta_{E,\Omega} + q_{E,\Omega})$ the operator acting by the rule $(-\Delta_E + q_e) u_e$, whose domain $D(-\Delta_{E,\Omega} + q_{E,\Omega})$ is the set of all the functions
$u = \{ u_e \}_{e \in E_{\Omega}} \in H^2(E_{\Omega})$ satisfying the Dirichlet condition $u(v) = 0$ at any boundary vertex $v \in \partial \Omega$ and the Kirchhoff conditions at any interior vertex $v \in \mathring{\Omega}$. By the standard argument, the operator $(-\Delta_{E,\Omega} + q_{E,\Omega})$ is self-adjoint and its spectrum $\sigma(-\Delta_{E,\Omega} + q_{E,\Omega})$ is a countable set of real eigenvalues. Below, we assume that
\begin{equation} \label{cond1}
\lambda \not\in \sigma(-\Delta_{E,\Omega} + q_{E,\Omega}).
\end{equation}

\textit{The edge Dirichlet-to-Neumann (D-N) map} $\Lambda_{E,\Omega}(\lambda)$ is defined as follows:
\begin{equation}
	\Lambda_{E,\Omega}(\lambda)f(v)=\partial_e u_{e}(v), \quad v\in \partial \Omega, \label{E_DN}
\end{equation}
where $u$ is the solution of the boundary value problem \eqref{BVP1} for the boundary data $f$ and $e$ is the edge of $E_{\Omega}$ having $v$ as its end point.
Below we assume that the vertex set $\Omega$ is fixed and use the short notation $\Lambda_E := \Lambda_{E,\Omega}$. 
Note that $f = \{ f(v)\}_{v \in \partial \Omega}$ can be treated as a vector of $\mathbb C^M$, $M := |\partial \Omega|$, and $\Lambda_E(\lambda)$, as the $(M \times M)$ matrix function such that $\Lambda_E(\lambda) f = g$, $g = \left\{ \frac{d}{dz}u_{e}(v) \right\}_{v \in \partial \Omega} \in \mathbb C^M$. The matrix function $\Lambda_E(\lambda)$ is analytic in $\lambda$ satisfying \eqref{cond1}.

For simplicity, we consider the domain 
$$
\mathring{\Omega} = D_N := \{ n_1 + \mathrm{i} n_2 \colon 0 \le n_1, n_2 \le N\},
$$
where a natural number $N$ is chosen so large that $\mbox{supp} \, q \subseteq \mathring{E}_{\Omega}$. In other words, $q_e = 0$ on all the edges except $e \in \mathring{E}_{\Omega}$. In Section~\ref{sec:scattering}, we show that the D-N map is uniquely specified by the S-matrix and vice versa, so the form of the domain is actually unimportant.

Consider Inverse Problem~\ref{ip:DNe} for the domain $\Omega$.
In fact, we have to find the potential on $\mathring{E}_{\Omega}$, since on all the other edges $q_e = 0$.
Our main result is the following uniqueness theorem.

\begin{theorem} \label{thm:edgeDN}
Suppose that the potential $q = \{ q_e\}_{e \in E}$ on the square lattice $(V, E)$ satisfies (Q-1), (Q-2), and (Q-3) and $\Lambda_E(\lambda)$ is the edge Dirichlet-to-Neumann map for the region $\Omega$ such that $\mathring\Omega = D_N$ and $\mbox{supp} \, q \subseteq \mathring{E}_{\Omega}$. Then
$\Lambda_E(\lambda)$ uniquely specifies the potential $q$. 
\end{theorem}

\section{Reduced vertex Laplacian}\label{sec:vertex}

The proof of Theorem~\ref{thm:edgeDN} is based on the reduction of the edge Laplacian to the vertex one.

Consider the Schr\"odinger equation
\begin{equation}
\left(-\frac{d^{2}}{dz^{2}}+q_{e}(z)-\lambda\right)\phi=0, \quad z \in (0,1), \label{D-D}
\end{equation}
on each fixed edge $e \in E$. Let $\phi_{e0}(z, \lambda), \phi_{e1}(z, \lambda)$ be the solutions of \eqref{D-D}
with the initial data $\phi_{e0}(0,\lambda)=0, \, \phi_{e0}'(0,\lambda)=1$ and $\phi_{e1}(1,\lambda)=0, \, \phi_{e1}'(1,\lambda)=-1$, respectively. 

Denote by $H_e$ the Schr\"odinger operator $\left(-\frac{d^{2}}{dz^{2}}+q_{e} \right)$ with the domain $D(H_e)$ which consists of functions $u \in H^2[0,1]$ satisfying the Dirichlet conditions $u(0) = u(1) = 0$. It is well-known that the operator $H_e$ is self-adjoint and its spectrum is a countable set of real eigenvalues.
In the following, we assume that $\lambda \not\in \sigma(H_e)$, $e \in E$.
This guarantees that $\phi_{e0}(1, \lambda) \neq 0$ and $\phi_{e1}(0, \lambda) \neq 0$. 

If $w, v \in V$ are two end points of an edge $e \in E$, then we denote
$$
\psi_{wv}(z, \lambda) =\begin{cases}
	\phi_{e0}(z, \lambda), &  e(0) = v,\\
	\phi_{e1}(1-z, \lambda), &  e(1) = v.
\end{cases} 
$$
Note that, by the assumption (Q-3), we have $\phi_{e0}(z, \lambda) = \phi_{e1}(1 - z, \lambda)$, hence $\psi_{wv}(1, \lambda) = \psi_{vw}(1, \lambda)$. In particular, if $q_e = 0$, then $\psi_{wv}(z, \lambda) = \frac{\sin \sqrt{\lambda}z}{\sqrt{\lambda}}$.
	
We define the reduced vertex Laplacian $\Delta_{V,\lambda}$ on $V$ by
\begin{equation}
	(\Delta_{V,\lambda}u)(v)=\frac{1}{4}\sum_{w\sim v}\frac{1}{\psi_{wv}(1,\lambda)}u(w), \quad v\in V,   \label{RVL}
\end{equation}
for $ u \in L^{2}_{loc}(V)$. We also define the scalar multiplication operator:
\begin{equation}
	(Q_{V,\lambda}u)(v)=q_{v,\lambda} u(v), \quad q_{v,\lambda} := \frac{1}{4}\sum_{w \sim v}\frac{\psi_{wv}'(1,\lambda)}{\psi_{wv}(1,\lambda)}. \label{Q_V}
\end{equation}

\begin{lemma} \label{lem:EV}
Let $v$ be a fixed vertex. If a function $u = \{ u_e \}_{e \in E_v} \in H^2(E_v)$ satisfies equation \eqref{D-D} for all $e \in E_v$ and the Kirchhoff conditions in $v$, then the relation $(-\Delta_{V,\lambda} + Q_{V,\lambda})u = 0$ holds at $v$.	
\end{lemma}

\begin{proof}
	Consider the solutions $S_{e}(z,\lambda)$ and $C_{e}(z,\lambda)$ of equation~of \eqref{D-D} satisfying the initial conditions $S_{e}(0,\lambda)=C'_{e}(0,\lambda)=0$ and $S'_{e}(0,\lambda)=C_e(0,\lambda) =1$. Then, any solution $u_{e}$ of \eqref{D-D}  can be written as 
\begin{equation}
	u_{e}(z
	)=a_{e}S_{e}(z,\lambda)+b_{e}C_{e}(z,\lambda)\label{u_SC}
\end{equation}
	where $a_{e}$ and $b_{e}$ are constants. 
	
	Let $v \in V$ be fixed. For convenience, suppose that $e(0) = v$ for all $e \in E_v$. Then \eqref{u_SC} together with (K-1) imply (see Figure~\ref{2}):
	\begin{figure}[h]
		\centering
		\includegraphics[width=5.5cm,height=4cm]{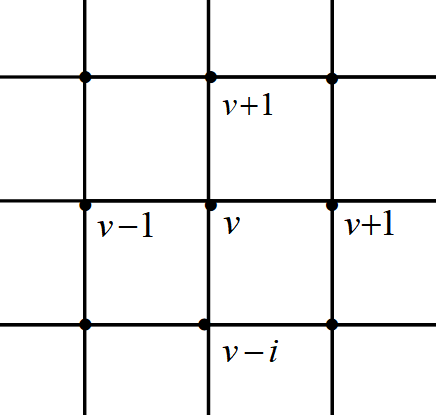}
		\caption{}
		\label{2}
	\end{figure}
	$$
	b_{(v,v+1)}=b_{(v,v-1)}=b_{(v,v+i)}=b_{(v,v-i)}=u(v),
	$$
	and (K-2) implies
	$$
	a_{(v,v+1)}+a_{(v,v-1)}+a_{(v,v+i)}+a_{(v,v-i)}=0.
	$$
	By \eqref{u_SC}, we have 
	$$
	a_{e}=\frac{u_{e}(1)}{S_{e}(1,\lambda)}-\frac{C_{e}(1,\lambda)}{S_{e}(1,\lambda)}u(v).
	$$
	Thus,
\begin{equation}
	\sum_{k=1,-1,i,-i}a_{(v,v+k)}=\sum_{k=1,-1,i,-i}\frac{u_{(v,v+k)}(1)}{S_{(v,v+k)}(1,\lambda)}-\frac{C_{(v,v+k)}(1,\lambda)}{S_{v,v+k}(1,\lambda)}u_{e}(0)=0.\label{suma_k}
\end{equation}
	Note that $S_{e}(1,\lambda)=\psi_{wv}(1,\lambda)$ and, under the assumption (Q-3), $C_{e}(1,\lambda)=S'_{e}(1,\lambda)$. Consequently, we arrive at the relation $(-\Delta_{V,\lambda} + Q_{V,\lambda})u(v) = 0$.
\end{proof}

Consider the vertex set $\Omega$ defined by \eqref{defOmega} for $\mathring{\Omega} = D_N$ and the interior boundary value problem
\begin{equation}
\begin{cases}
(-\Delta_{V,\lambda}+Q_{V,\lambda})u=0\ in\  \mathring{\Omega},\\
u=f\ on \ \partial \Omega.
\end{cases}\label{BVP2}
\end{equation}

By virtue of Lemma~\ref{lem:EV}, if a function $u = \{ u_e \}_{e \in E_{\Omega}}$ solves the edge boundary value problem \eqref{BVP1}, then its values $\{ u(v) \}_{v \in \Omega}$ in the vertices satisfy \eqref{BVP2}.
    
Define the vertex degree in the subgraph $(\Omega, E_{\Omega})$ as follows:
$$
\mbox{deg}_{\Omega}(v)=\begin{cases}
   	\#\{w\in \Omega \colon w\sim v\}, & v \in \mathring{\Omega},\\
   	\#\{w\in \mathring{\Omega} \colon w\sim v\}, & v \in \partial \Omega.
\end{cases}
$$
    
Then, \textit{the vertex D-N map} is defined by
\begin{equation*}
   	\Lambda_{V, \Omega}(\lambda)f(v)=-\frac{1}{\mbox{deg}_{\Omega}(v)}\sum_{w\sim v,w\in \mathring{\Omega}} u(w), \quad v \in \partial \Omega. 
\end{equation*}

Below we use a shorter notation $\Lambda_V := \Lambda_{V, \Omega}$.
Note that, for the region $\Omega = D_N$, the degrees $\mbox{deg}_{\Omega}(v)$ of the boundary vertices $v \in \partial \Omega$ equal to $1$. Hence
\begin{equation}
\Lambda_V(\lambda)f(v)=-u(w), \quad v \in \partial \Omega, \quad w \in \mathring{\Omega}, \quad w \sim v. \label{V_DN}
\end{equation}

The edge and the vertex D-N maps are closely related to each other, which is shown in the following lemma.

\begin{lemma}\label{E-V_DN}
The following equality holds:
\begin{equation}
\Lambda_V(\lambda)=-\cos \sqrt{\lambda} I +\frac{\sin \sqrt{\lambda}}{\sqrt{\lambda}}\Lambda_E(\lambda). \label{relation-DN}
\end{equation}	
where $I$ is the unit operator in $\mathbb C^M$ and $\lambda$ satisfies \eqref{cond1}.
\end{lemma}

\begin{proof}
	Let $e(1)=v,\ v\in \partial\Omega$, then
	\begin{equation}
	\Lambda_V(\lambda)f(v)=-u_{e}(0),\quad \Lambda_E(\lambda)f(v)=u'_{e}(1).\label{u0u'1}
	\end{equation}
	By \eqref{u_SC}, we have
	\begin{equation}
	u_{e}(1)=a_{e}S_{e}(1,\lambda)+u_{e}(0)C_{e}(1,\lambda)=f(v),\label{ue1}
	\end{equation}
	and
	\begin{equation}
	u'_{e}(1)=a_{e}S'_{e}(1,\lambda)+u_{e}(0)C'_{e}(1,\lambda).\label{u'e1}
    \end{equation}
	Then
	\begin{equation}
	a_{e}=\frac{1}{S_{e}(1,\lambda)}(f(v)-u_{e}(0)C_{e}(1,\lambda)).\label{a_e}
    \end{equation}
	Substituting \eqref{a_e} into \eqref{u'e1} and using the relation 
	$$
	S_{e}(1,\lambda)C'_{e}(1,\lambda)-C_{e}(1,\lambda)S'_{e}(1,\lambda)=-1
	$$
	we obtain 
	\begin{equation}\label{d-n1}
	\begin{split}
	u'_{e}(1)
	&=f(v)\frac{S'_{e}(1,\lambda)}{S_{e}(1,\lambda)}+(C'_{e}(1,\lambda)-\frac{C_{e}(1,\lambda)S'_{e}(1,\lambda)}{S_{e}(1,\lambda)})u_{e}(0)\\
	&=f(v)\frac{S'_{e}(1,\lambda)}{S_{e}(1,\lambda)}-\frac{1}{S_{e}(1,\lambda)}u_{e}(0).
    \end{split}
    \end{equation}
	By \eqref{u0u'1} and \eqref{d-n1}, we get
	\begin{equation} \label{sm1}
	\Lambda_{V}(\lambda)f(v)=-f(v)S'_{e}(1,\lambda)+(\Lambda_{E}(\lambda)f(v))S_{e}(1,\lambda).
	\end{equation}
	Note that the potentials on the edge $e$ connected to the vertex $v\in \partial \Omega$ are zero, so we have
	$$
	S_e(1,\lambda)=\frac{\sin \sqrt{\lambda}}{\sqrt{\lambda}},\quad S'_e(1,\lambda)=\cos \sqrt{\lambda}.
	$$
	Using these relations together with \eqref{sm1}, we arrive at \eqref{relation-DN}.
	
	Now, let $e(0)=v,\ v\in \partial\Omega$, then
	\begin{equation}
	\Lambda_V(\lambda)f(v)=-u_{e}(1),\quad \Lambda_E(\lambda)f(v)=-u'_{e}(0).\label{u1u'0}
	\end{equation}
	By the boundary condition, we get 
	$u_{e}(0)=f(v)$. By \eqref{u_SC}, we have
	\begin{equation}
		u_{e}(1)=a_{e}S_{e}(1,\lambda)+f(v)C_{e}(1,\lambda),\label{ue11}
	\end{equation}
	and
	\begin{equation}
		u'_{e}(0)=a_{e}.\label{ae11}
	\end{equation}
    Substitute \eqref{ae11} into \eqref{ue11}, we obtain 
    \begin{equation} 
    u_{e}(1)=u'_{e}(0)S_{e}(1,\lambda)+f(v)C_{e}(1,\lambda).\label{d-n2}
    \end{equation}
    By \eqref{u1u'0} and \eqref{d-n2}, we have
    $$
    \Lambda_{V}(\lambda)f(v)=-f(v)C_{e}(1,\lambda)+(\Lambda_{E}(\lambda)f(v))S_{e}(1,\lambda).
    $$
    Note that 
    $$
    S_e(1,\lambda)=\frac{\sin \sqrt{\lambda}}{\sqrt{\lambda}},\quad C_e(1,\lambda)=\cos \sqrt{\lambda},
    $$
    so we directly derive \eqref{relation-DN}.
\end{proof}

Therefore, Inverse Problem~\ref{ip:DNe} of recovering the potential from the edge D-N map can be easily reduced to the following inverse problem.

\begin{ip} \label{ip:DNv}
	Given the vertex D-N map $\Lambda_V$, find the potential $q$.
\end{ip}

\section{Special solutions of the vertex Schr\"odinger equation}\label{sec:solution}
    We are now in a position to construct the solution of Inverse Problem~\ref{ip:DNv}. For this purpose, we first obtain some special solutions of the vertex Schr\"odinger equation 
 	$$
 		(-\Delta_{V, \lambda} + Q_{V, \lambda})u = 0.
 	$$  
    
    The boundary $\partial\Omega$ of the domain $\mathring{\Omega}=D_{N}$ consists of the four parts, see Figure~\ref{3}:
    \begin{figure}[h]
    	\centering
    	\includegraphics[width=6.5cm,height=5.5cm]{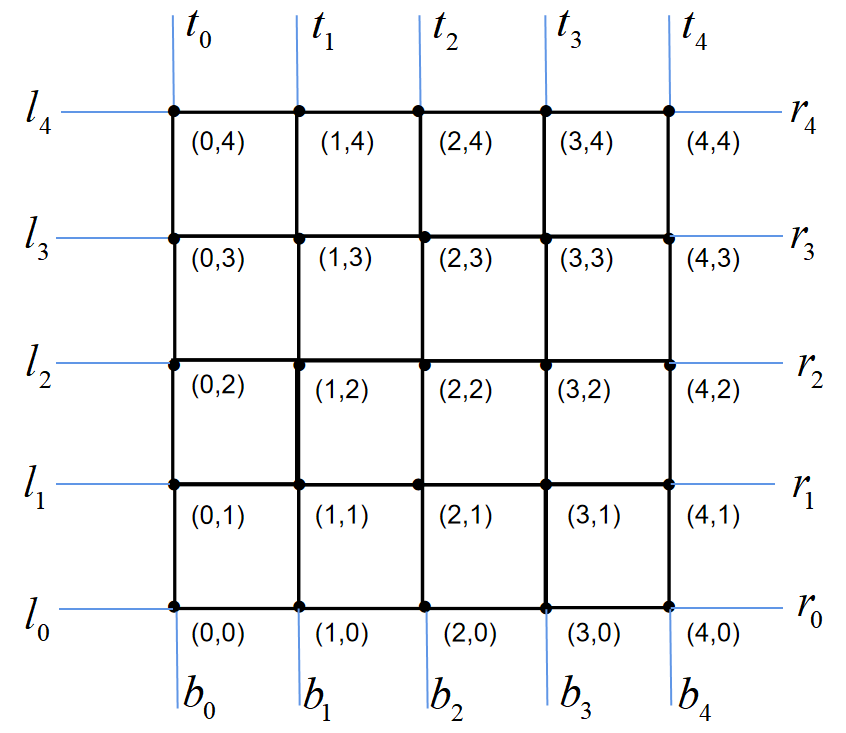}
    	\caption{N=4}
    	\label{3}
    \end{figure}
    \begin{align*}
    & (\partial\Omega)_{T}=\{ (N+1)\mathrm{i}+m \colon 0\leqslant m\leqslant N \}=\{t_{m}\}_{m=0}^N, \\
    & (\partial\Omega)_{B}=\{ -\mathrm{i}+m \colon 0\leqslant m\leqslant N \}=\{b_{m}\}_{m=0}^N, \\
    & (\partial\Omega)_{L}= \{ -1+m\mathrm{i} \colon 0\leqslant m\leqslant N \}=\{l_{m}\}_{m=0}^N, \\
    & (\partial\Omega)_{R}=\{ N+1+m\mathrm{i} \colon 0\leqslant m\leqslant N \}=\{r_{m}\}_{m=0}^N.
    \end{align*}
    
    Recall that $\Lambda_{V} = \Lambda_{V, \Omega}$ is the vertex D-N map defined by \eqref{V_DN}. The key to the inverse procedure is the following partial data problem.
    
    \begin{lemma}\label{SS1}
    	(1) Given a partial Dirichlet data $f$ on $\partial \Omega \backslash (\partial\Omega)_{R}$, and a partial
    	Neumann data $ g$ on $(\partial\Omega)_{L}$, there is a unique solution $u$ in $\Omega$ of the boundary value problem
    	\begin{equation}
    		\begin{cases}
    			(-\Delta_{V,\lambda}+Q_{V,\lambda})u=0\ in\  \mathring{\Omega},\\
    			u=f\ on \ \partial \Omega\backslash (\partial\Omega)_{R},\\
    			\partial_v^{\Omega}u=g\ on \ (\partial\Omega)_{L}.
    		\end{cases}\label{BVP3}
    	\end{equation}
    	where 
    	$$
    	\partial_v^{\Omega} u=-\frac{1}{\mbox{deg}_{\Omega}(v)}\sum_{w\sim v,w\in \mathring{\Omega}}u(w).
    	$$
    	
    	(2) Given the D-N map $\Lambda_{V}$, a partial Dirichlet data $f_{1}$ on $\partial \Omega \backslash (\partial\Omega)_{R}$ and a partial Neumann data $g$ on $(\partial\Omega)_{L}$, there exists a unique $f$ on $\partial \Omega$ such that $f=f_{1}$ on $\partial \Omega \backslash (\partial\Omega)_{R}$ and $\Lambda_{V}f=g$ on $(\partial\Omega)_{L}$.
    \end{lemma}
    
    \begin{proof}
    	(1) The values of $u(m\mathrm{i}),\  m=0,\cdots,N$ are computed from the values of $g$. Using equation (\ref{BVP3}), the Dirichlet data  $f$ and the known values of $u$, one can then compute $u(1+m\mathrm{i}),\ m=0,\cdots,N$. Next, we obtain the values of $u(2+m\mathrm{i}),\ m=0,\cdots,N$. Repeating this procedure, we get $u(v)$ for all $v \in \Omega$.
    	
    	(2) For subsets $A, B \subset \partial\Omega$, we denote the associated submatrix of $\Lambda_{V}$ by
    	$\Lambda_{V}(A;B)$. Suppose $f_{2}=0$ on $\partial\Omega\backslash (\partial\Omega)_{R}$ and $\Lambda_{V}f_{2}=0$
    	on $(\partial\Omega)_{L}$. By (1) and the boundary conditions, the solution
    	$u$ vanishes in $\Omega$ identically. Hence $f_{2} = 0$ on $(\partial\Omega)_{R}$. This implies the submatrix
    	$$
    	\Lambda_{V}((\partial\Omega)_{L};(\partial\Omega)_{R}):(\partial\Omega)_{R} \longrightarrow (\partial\Omega)_{L}
    	$$
    	is a bijection. We seek $f$ in the form
    	$$
    	(\Lambda_{V}f)\big|_{(\partial\Omega)_{L}}=\Lambda_{V}((\partial\Omega)_{L};(\partial\Omega)_{R})f_{3}+\Lambda_{V}((\partial\Omega)_{L};\partial\Omega\backslash(\partial\Omega)_{R})f_{1}=g,
    	$$
    	where
    	$$
    	f_{3}=\left( \Lambda_{V}((\partial\Omega)_{L};(\partial\Omega)_{R})\right) ^{-1}\left( g-\Lambda_{V}((\partial\Omega)_{L};\partial\Omega\backslash(\partial\Omega)_{R})f_{1} \right).
    	$$
    	This proves the lemma.
    \end{proof}
    
    Now, for $N+1 \leqslant k \leqslant 2N$, let us consider the diagonal line
    \begin{equation}
    	A_{k} = \{x_{1} + \mathrm{i}x_{2} \colon x_{1} +x_{2} = k \}.\label{A_k}
    \end{equation}
    The vertices on $A_{k} \cap \Omega$ are written as
    \begin{equation}
    	\alpha_{k,l} = \alpha_{k,0} + l(1-\mathrm{i}), \quad l=0,1,2,\dots,2N+2-k,\label{a_kl}
    \end{equation}
    where $\alpha_{k,0}=t_{k-(N+1)}=k-(N+1)+\mathrm{i}(N+1)$.
    \begin{lemma}\label{SS2}
    	Let $A_{k} \cap \partial\Omega = \{\alpha_{k,0}, \alpha_{k,m}\},k=N+1,\cdots,2N$. Then, there exists a unique solution $u$ in $\Omega$ of the boundary value problem
    	\begin{equation}
    		(-\Delta_{V,\lambda}+Q_{V,\lambda})u=0\ in\  \mathring{\Omega},\label{Eq1}
    	\end{equation}
    	with partial Dirichlet data $f$ such that
    	\begin{equation}
    		\begin{cases}
    			f(\alpha_{k,0})=1,\\
    			f(z)=0,\ for \ z\in \partial\Omega\backslash \left( (\partial\Omega)_{R}\cup \alpha_{k,0} \right) 
    		\end{cases}\label{BV}
    	\end{equation}
    	and partial Neumann data $g = 0$ on $(\partial\Omega)_{L}$. It satisfies
    	\begin{equation}
    		u(x_{1}+\mathrm{i}x_{2})=0\ if \ x_{1}+x_{2} < k.\label{u}
    	\end{equation}
    \end{lemma}
    \begin{proof}
    	The uniqueness and the existence of $u$ on $\Omega$ follow from Lemma \ref{SS1}. By the condition on $f,g$, one can compute $u(x_{1} + \mathrm{i}x_{2})$ successively to obtain \eqref{u}.
    \end{proof}
    
    An important feature of the solution $u$ in Lemma~\ref{SS2} is that $u$ vanishes below the line $A_{k}$. For such solutions, we obtain the following property. 
    
    Let $u$ be a solution of the equation
    \begin{equation}
    	(-\Delta_{V,\lambda}+Q_{V,\lambda})u=0\ in\  \mathring{\Omega},\label{Eq2}
    \end{equation}
    which vanishes in the region $x_{1} +x_{2} < k$. Let $a, b,b',c,d \in V$ and $e, e' \in E$ be as in Figure~\ref{4}.
    \begin{figure}[h]
    	\centering
    	\includegraphics[width=3.5cm,height=2.5cm]{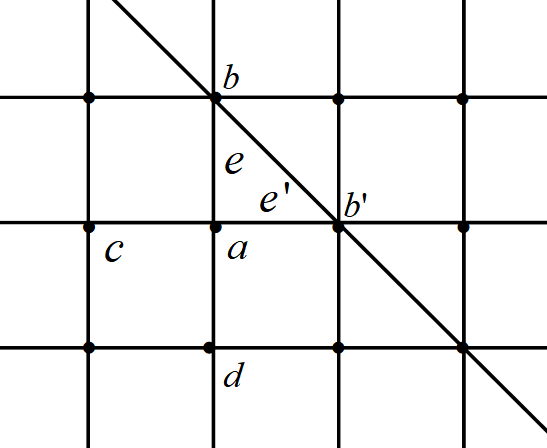}
    	\caption{}
    	\label{4}
    \end{figure}
    
    Then, evaluating the equation (\ref{Eq2}) at $v = a$ and using (\ref{RVL}), (\ref{Q_V}), we obtain
    \begin{equation}
    	\frac{1}{\psi_{ba}(1,\lambda)}u(b)+\frac{1}{\psi_{b'a}(1,\lambda)}u(b')=0.\label{Eq3}
    \end{equation}
    
    The special solutions of Lemma~\ref{SS2} and their property \eqref{Eq3} play a crucial role in the proof of the main theorem and in the reconstruction procedure in the next section.
    
\section{Reconstruction procedure} \label{sec:construct}
	
	The goal of this section is to prove Theorem~\ref{thm:edgeDN} on the uniqueness of the inverse spectral problem solution. The proof consists in a reconstruction procedure, which is based on the reduction to the discrete inverse problem by the vertex D-N map (i.e. Inverse Problem~\ref{ip:DNv}), on specific properties of the square lattice, and on applying the classical results for the recovery of the Schr\"odinger potential on each edge.
	
    Before proceeding to the reconstruction, we provide two well-known results for inverse spectral problems on a finite interval.
    Let $e = (w,v)$ be a fixed edge. Clearly, the eigenvalues of the Dirichlet problem
    $$
    	\left( -\frac{d^2}{dz^2} + q_e(z) - \lambda \right) u = 0, \quad u(0) = u(1) = 0
    $$
    coincide with the zeros of the characteristic function $\psi_{wv}(1, \lambda)$.
    
    \begin{lemma}{(\cite[Theorem 1.4.3]{Fre})} \label{Th143}
    	 If the potential $q_{e}$ is symmetric, then the Dirichlet eigenvalues for the operator $-\frac{d^{2}}{dz^{2}}+q_{e}(z)$ on $(0, 1)$ uniquely determine the potential $q_{e}$. In other words, $q_e$ is uniquely specified by the function $\psi_{wv}(1, \lambda)$.
    \end{lemma}
    
    \begin{lemma}{(\cite[Theorem 1.4.7]{Fre})} \label{Th147}
    	 The Weyl function $\frac{\psi'_{wv}(1,\lambda)}{\psi_{wv}(1,\lambda)}$ uniquely specifies the potential $q_e$.
    \end{lemma}

	The both inverse problems of Lemmas~\ref{Th143} and~\ref{Th147} can be solved constructively by the well-known methods, e.g., the Gelfand-Levitan method and the method of spectral mappings (see \cite{Fre}).
    
    Now, let us prove Theorem~\ref{thm:edgeDN} by providing a
    reconstruction procedure for solving Inverse Problem~\ref{ip:DNe}.
 
	\medskip
	   
    \textbf{Reconstruction procedure}. Suppose that the potential $q$ on the square lattice fulfills the conditions (Q-1), (Q-2), and (Q-3). Let $\Omega$ be the square region defined in Section~\ref{sec:edge} and let the corresponding edge D-N map $\Lambda_E$ be given.
    
    \smallskip
    
    \textbf{Step 1.} We obtain the vertex D-N map  $\Lambda_V$ by the formula \eqref{relation-DN}.
    
    \smallskip
    
    \textbf{Step 2.} For $k = 2N, 2N-1, \dots, N+1$, implement the following steps 3--5.
    
    \smallskip
    
    \textbf{Step 3.} For the value of $k$ fixed at step 2, draw the line $A_{k}$ defined by (\ref{A_k}) and take the boundary data $f$ having the properties in Lemma \ref{SS2}. Under the assumption (Q-2), we have $q_{e}(z)=0$ on all the boundary edges $\partial E_{\Omega}$. So, we know the functions 
    $\psi_{wv}(1,\lambda)=\frac{\sin \sqrt{\lambda}}{\sqrt{\lambda}}$
    for $(w,v) \in \partial E_\Omega$.
    By Lemma \ref{SS1} (2), we can find $u$ on $(\partial\Omega)_{R}$ with the vertex D-N map $\Lambda_{V}$. Thus, we know $u$ on $\partial\Omega$.

    \smallskip
    
    \textbf{Step 4.} If $k=2N$, using \eqref{V_DN} and the value of $u(\alpha_{k,1}+\mathrm{i})$, we find
    $$
    u(\alpha_{k,1})=-\Lambda_V u(\alpha_{k,1}+\mathrm{i}),
    $$
    By (\ref{Eq3}), we have 
    \begin{equation}
    	\frac{u(\alpha_{k,0})}{\psi_{\alpha_{k,0}-\mathrm{i},\alpha_{k,0}}(1,\lambda)}+\frac{u(\alpha_{k,1})}{\psi_{\alpha_{k,1}-1,\alpha_{k,1}}(1,\lambda)}=0,
    \end{equation}
    and 
    \begin{equation}
    	\frac{u(\alpha_{k,1})}{\psi_{\alpha_{k,1}-\mathrm{i},\alpha_{k,1}}(1,\lambda)}+\frac{u(\alpha_{k,2})}{\psi_{\alpha_{k,2}-1,\alpha_{k,2}}(1,\lambda)}=0.
    \end{equation}
    By the known values $u(\alpha_{k,0})$, $u(\alpha_{k,1})$, $u(\alpha_{k,2})$, $ \psi_{\alpha_{k,1}-\mathrm{i},\alpha_{k,1}}(1,\lambda)$, $ \psi_{\alpha_{k,2}-1,\alpha_{k,2}}(1,\lambda)$, we can find
    \begin{equation}
    	\psi_{\alpha_{k,1}-1,\alpha_{k,1}}(1,\lambda)=-\frac{u(\alpha_{k,1})}{u(\alpha_{k,0})}\psi_{\alpha_{k,0}-\mathrm{i},\alpha_{k,0}}(1,\lambda),\label{psi1}
    \end{equation} 
    and  
    \begin{equation}
    	\psi_{\alpha_{k,1}-\mathrm{i},\alpha_{k,1}}(1,\lambda)=-\frac{u(\alpha_{k,1})}{u(\alpha_{k,2})}\psi_{\alpha_{k,2}-1,\alpha_{k,2}}(1,\lambda).\label{psi2}
    \end{equation}  
    Note that the zeros of $\psi_{\alpha_{k,1}-1,\alpha_{k,1}}(1,\lambda)$ and $ \psi_{\alpha_{k,1}-\mathrm{i},\alpha_{k,1}}(1,\lambda)$ are the Dirichlet eigenvalues for the operator $-\frac{d^{2}}{dz^{2}}+q_{e}(z)$ on $(0, 1)$. Since the potential is symmetric, by Lemma \ref{Th143}, these eigenvalues uniquely determine the potentials $q_{e}(z)$ on $(\alpha_{k,1}-1,\alpha_{k,1})$ and $(\alpha_{k,1}-\mathrm{i},\alpha_{k,1})$.
    
    \smallskip
    
    \textbf{Step 5.} If $k\le 2N-1$, then implement steps 5.1--5.6.
    
    \textbf{Step 5.1} Using \eqref{V_DN} and the values of $u(\alpha_{k,1}+\mathrm{i})$,  $u(\alpha_{k,2N+1-k}+1)$, we find
    $$
    u(\alpha_{k,1})=-\Lambda_V u(\alpha_{k,1}+\mathrm{i}),\quad u(\alpha_{k,2N+1-k})=-\Lambda_V u(\alpha_{k,2N+1-k}+1).
    $$ 
    By (\ref{Eq3}), we get 
    \begin{equation}
    	\psi_{\alpha_{k,1}-1,\alpha_{k,1}}(1,\lambda)=-\frac{u(\alpha_{k,1})}{u(\alpha_{k,0})}\psi_{\alpha_{k,0}-\mathrm{i},\alpha_{k,0}}(1,\lambda),\label{psi3}
    \end{equation} 
    and 
    \begin{equation}
    	\psi_{\alpha_{k,2N+1-k}-\mathrm{i},\alpha_{k,2N+1-k}}(1,\lambda)=-\frac{u(\alpha_{k,2N+1-k})}{u(\alpha_{k,2N+2-k})}\psi_{\alpha_{k,2N+2-k}-1,\alpha_{k,2N+2-k}}(1,\lambda).\label{psi4}
    \end{equation} 
    Then, by Lemma \ref{Th143}, the zeros of $\psi_{\alpha_{k,1}-1,\alpha_{k,1}}(1,\lambda)$ and $\psi_{\alpha_{k,2N+1-k}-\mathrm{i},\alpha_{k,2N+1-k}}(1,\lambda)$ uniquely determine the potentials on $(\alpha_{k,1}-1,\alpha_{k,1})$ and $(\alpha_{k,2N+1-k}-\mathrm{i},\alpha_{k,2N+1-k})$.
    
    \textbf{Step 5.2.} Evaluating equation (\ref{Eq2}) at $\alpha_{k,1}$, we obtain
    $$
    -\frac{1}{4}\left( \frac{u(\alpha_{k,1}+\mathrm{i})}{\psi_{\alpha_{k,1},\alpha_{k,1}+\mathrm{i}}(1,\lambda)}+\frac{u(\alpha_{k,1}+1)}{\psi_{\alpha_{k,1},\alpha_{k,1}+1}(1,\lambda)} \right)+q_{v,\lambda}(\alpha_{k,1})u(\alpha_{k,1})=0.
    $$
    Then, we can get the value of $q_{v,\lambda}(\alpha_{k,1})$. By (\ref{Q_V}), we have 
    \begin{align*}
    	q_{v,\lambda}(\alpha_{k,1})&=\frac{1}{4}\Biggl( \frac{\psi'_{\alpha_{k,1}-1,\alpha_{k,1}}(1,\lambda)}{\psi_{\alpha_{k,1}-1,\alpha_{k,1}}(1,\lambda)}+\frac{\psi'_{\alpha_{k,1},\alpha_{k,1}+1}(1,\lambda)}{\psi_{\alpha_{k,1},\alpha_{k,1}+1}(1,\lambda)}\\
    	&+\frac{\psi'_{\alpha_{k,1}-\mathrm{i},\alpha_{k,1}}(1,\lambda)}{\psi_{\alpha_{k,1}-\mathrm{i},\alpha_{k,1}}(1,\lambda)}+\frac{\psi'_{\alpha_{k,1},\alpha_{k,1}+\mathrm{i}}(1,\lambda)}{\psi_{\alpha_{k,1},\alpha_{k,1}+\mathrm{i}}(1,\lambda)} \Biggr)
    \end{align*}
    where the values of $\psi_{\alpha_{k,1}-1,\alpha_{k,1}}(1,\lambda)$, $\psi_{\alpha_{k,1},\alpha_{k,1}+\mathrm{i}}(1,\lambda)$ and $\psi_{\alpha_{k,1},\alpha_{k,1}+1}(1,\lambda)$ are known. 
    
    \textbf{Step 5.3.} Obtain the value 
    \begin{align}
    	\frac{\psi'_{\alpha_{k,1}-\mathrm{i},\alpha_{k,1}}(1,\lambda)}{\psi_{\alpha_{k,1}-\mathrm{i},\alpha_{k,1}}(1,\lambda)}=4q_{v,\lambda}(\alpha_{k,1})-\frac{\psi'_{\alpha_{k,1}-1,\alpha_{k,1}}(1,\lambda)}{\psi_{\alpha_{k,1}-1,\alpha_{k,1}}(1,\lambda)}\notag\\
    	-\frac{\psi'_{\alpha_{k,1},\alpha_{k,1}+1}(1,\lambda)}{\psi_{\alpha_{k,1},\alpha_{k,1}+1}(1,\lambda)}-\frac{\psi'_{\alpha_{k,1},\alpha_{k,1}+\mathrm{i}}(1,\lambda)}{\psi_{\alpha_{k,1},\alpha_{k,1}+\mathrm{i}}(1,\lambda)}\label{p'p1}
    \end{align}
    which is the Weyl function associated with the potential on $(\alpha_{k,1}-\mathrm{i},\alpha_{k,1})$. 
    
    \textbf{Step 5.4.} Similarly to steps 5.2--5.3, evaluating equation (\ref{Eq2}) at $\alpha_{k,2N+2-k}$, we get the Weyl function  
    \begin{align}
    	\frac{\psi'_{\alpha_{k,2N+1-k}-1,\alpha_{k,2N+1-k}}(1,\lambda)}{\psi_{\alpha_{k,2N+1-k}-1,\alpha_{k,2N+1-k}}(1,\lambda)}=4q_{v,\lambda}(\alpha_{k,2N+1-k})-\frac{\psi'_{\alpha_{k,2N+1-k}-\mathrm{i},\alpha_{k,2N+1-k}}(1,\lambda)}{\psi_{\alpha_{k,2N+1-k}-\mathrm{i},\alpha_{k,2N+1-k}}(1,\lambda)}\notag\\
    	-\frac{\psi'_{\alpha_{k,2N+1-k},\alpha_{k,2N+1-k}+1}(1,\lambda)}{\psi_{\alpha_{k,2N+1-k},\alpha_{k,2N+1-k}+1}(1,\lambda)}-\frac{\psi'_{\alpha_{k,2N+1-k},\alpha_{k,2N+1-k}+\mathrm{i}}(1,\lambda)}{\psi_{\alpha_{k,2N+1-k},\alpha_{k,2N+1-k}+\mathrm{i}}(1,\lambda)}\label{p'p2}
    \end{align}
    associated to the potential on $(\alpha_{k,2N+1-k}-1,\alpha_{k,2N+1-k})$.
    
    \textbf{Step 5.5.} Repeating the procedure analogous to steps 5.2--5.4, we can obtain the values of $\frac{\psi'_{\alpha_{k,l}-1,\alpha_{k,l}}(1,\lambda)}{\psi_{\alpha_{k,l}-1,\alpha_{k,l}}(1,\lambda)},\ l=2,\cdots,2N+1-k$ and $\frac{\psi'_{\alpha_{k,l}-\mathrm{i},\alpha_{k,l}}(1,\lambda)}{\psi_{\alpha_{k,l}-\mathrm{i},\alpha_{k,l}}(1,\lambda)},\ l=1,\cdots,2N-k$. 
    
    \textbf{Step 5.6.} The potentials on $(\alpha_{k,l}-1,\alpha_{k,l}) ,\ l=2,\cdots,2N+1-k$ and $(\alpha_{k,l}-\mathrm{i},\alpha_{k,l}),\ l=1,\cdots,2N-k$ are uniquely determined by the Weyl functions by Lemma \ref{Th147}.
    
    Thus, we have constructed all the potentials on the upper triangular region of the square domain.
    
    \smallskip
    
    \textbf{Step 6.} Rotate the whole system by the angle $\pi$ and take a square domain congruent to the previous one. Repeat the steps 2--5 to determine the potentials on all the remaining edges of $E_{\Omega}$.
    
    \medskip
    
    This reconstruction procedure determines the potential uniquely on each edge, so it implies the proof of Theorem~\ref{thm:edgeDN}.
    
    Here we give a simple example of the reconstruction for $N=3$. 
    
    As in Figure \ref{5}, by step 4, we get $u(\alpha_{6,1})$ by the D-N map. By (\ref{psi1}) and (\ref{psi2}), we obtain
    $$
    \psi_{\alpha_{6,1}-1,\alpha_{6,1}}(1,\lambda)=-\frac{u(\alpha_{6,1})}{u(\alpha_{6,0})}\psi_{\alpha_{6,0}-\mathrm{i},\alpha_{6,0}}(1,\lambda),
    $$
    $$
    \psi_{\alpha_{6,1}-\mathrm{i},\alpha_{6,1}}(1,\lambda)=-\frac{u(\alpha_{6,1})}{u(\alpha_{6,2})}\psi_{\alpha_{6,2}-1,\alpha_{6,2}}(1,\lambda).
    $$
    Then, the potentials on $e_{1},e_{2}$ are determined by Lemma \ref{Th143}.
    \begin{figure}[h]
    	\centering
    	\includegraphics[width=5.5cm,height=5cm]{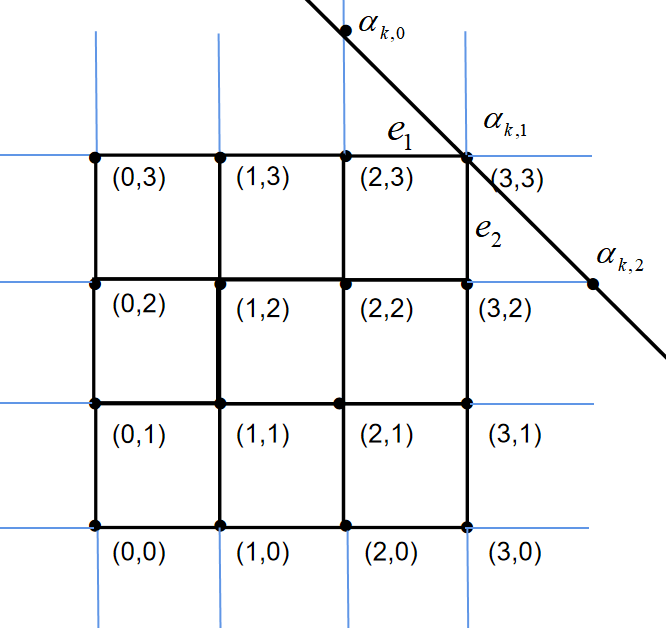}
    	\caption{N=3,k=6}
    	\label{5}
    \end{figure}
    
    As in Figure \ref{6}, by (\ref{psi3}),(\ref{psi4}), we obtain 
    $$
    \psi_{\alpha_{5,1}-1,\alpha_{5,1}}(1,\lambda)=-\frac{u(\alpha_{5,1})}{u(\alpha_{5,0})}\psi_{\alpha_{5,0}-\mathrm{i},\alpha_{5,0}}(1,\lambda),
    $$
    $$
    \psi_{\alpha_{5,2}-\mathrm{i},\alpha_{5,2}}(1,\lambda)=-\frac{u(\alpha_{5,2})}{u(\alpha_{5,3})}\psi_{\alpha_{5,3}-1,\alpha_{5,3}}(1,\lambda).
    $$
    Then, the potentials on $e_{3},e_{6}$ are determined.
    By (\ref{p'p1}) and (\ref{p'p2}), we have
    \begin{align*}
    	\frac{\psi'_{\alpha_{5,1}-\mathrm{i},\alpha_{5,1}}(1,\lambda)}{\psi_{\alpha_{5,1}-\mathrm{i},\alpha_{5,1}}(1,\lambda)}=4q_{v,\lambda}(\alpha_{5,1})-\frac{\psi'_{\alpha_{5,1}-1,\alpha_{5,1}}(1,\lambda)}{\psi_{\alpha_{5,1}-1,\alpha_{5,1}}(1,\lambda)}\notag\\
    	-\frac{\psi'_{\alpha_{5,1},\alpha_{5,1}+1}(1,\lambda)}{\psi_{\alpha_{5,1},\alpha_{5,1}+1}(1,\lambda)}-\frac{\psi'_{\alpha_{5,1},\alpha_{5,1}+\mathrm{i}}(1,\lambda)}{\psi_{\alpha_{5,1},\alpha_{5,1}+\mathrm{i}}(1,\lambda)},
    \end{align*}
    and 
    \begin{align*}
    	\frac{\psi'_{\alpha_{5,2}-1,\alpha_{5,2}}(1,\lambda)}{\psi_{\alpha_{5,2}-1,\alpha_{5,2}}(1,\lambda)}=4q_{v,\lambda}(\alpha_{5,2})-\frac{\psi'_{\alpha_{5,2}-\mathrm{i},\alpha_{5,2}}(1,\lambda)}{\psi_{\alpha_{5,2}-\mathrm{i},\alpha_{5,2}}(1,\lambda)}\notag\\
    	-\frac{\psi'_{\alpha_{5,2},\alpha_{5,2}+1}(1,\lambda)}{\psi_{\alpha_{5,2},\alpha_{5,2}+1}(1,\lambda)}-\frac{\psi'_{\alpha_{5,2},\alpha_{5,2}+\mathrm{i}}(1,\lambda)}{\psi_{\alpha_{5,2},\alpha_{5,2}+\mathrm{i}}(1,\lambda)}.
    \end{align*}
    Then, we can determine the potentials on $e_{4},e_{5}$ by Lemma \ref{Th147}.
    
    \begin{figure}[h]
    	\centering
    	\includegraphics[width=5.5cm,height=5cm]{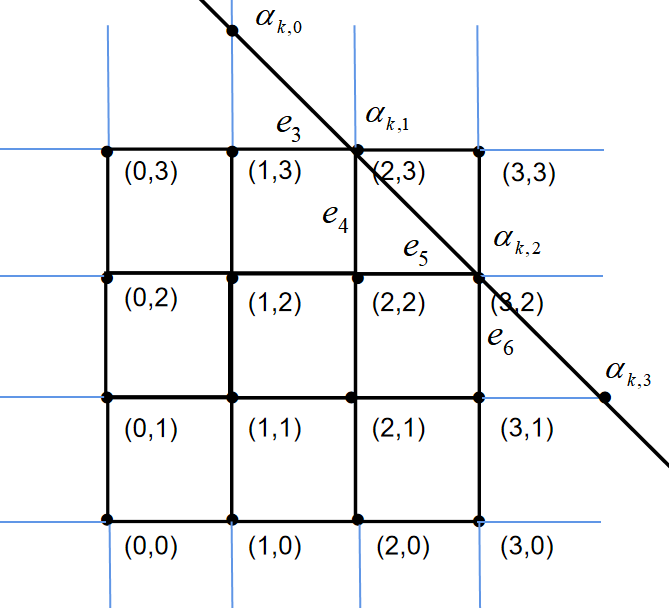}
    	\caption{N=3,k=5}
    	\label{6}
    \end{figure}
    
    As in Figure \ref{7}, by (\ref{psi3}),(\ref{psi4}) in step 5, we get 
    $$
    \psi_{\alpha_{4,1}-1,\alpha_{4,1}}(1,\lambda)=-\frac{u(\alpha_{4,1})}{u(\alpha_{4,0})}\psi_{\alpha_{4,0}-\mathrm{i},\alpha_{4,0}}(1,\lambda),
    $$
    $$
    \psi_{\alpha_{4,3}-\mathrm{i},\alpha_{4,3}}(1,\lambda)=-\frac{u(\alpha_{4,3})}{u(\alpha_{4,4})}\psi_{\alpha_{4,4}-1,\alpha_{4,4}}(1,\lambda).
    $$
    Then, by using Lemma \ref{Th143}, the potentials on $e_7,\ e_{12}$ are specified. 
    
    By (\ref{p'p1}),(\ref{p'p2}) and the argument in step 5, we get the values of   
    $\frac{\psi'_{\alpha_{4,1}-\mathrm{i},\alpha_{4,1}}(1,\lambda)}{\psi_{\alpha_{4,1}-\mathrm{i},\alpha_{4,1}}(1,\lambda)}$, $\frac{\psi'_{\alpha_{4,3}-1,\alpha_{4,3}}(1,\lambda)}{\psi_{\alpha_{4,3}-1,\alpha_{4,3}}(1,\lambda)}$, $\frac{\psi'_{\alpha_{4,2}-\mathrm{i},\alpha_{4,2}}(1,\lambda)}{\psi_{\alpha_{4,2}-\mathrm{i},\alpha_{4,2}}(1,\lambda)}$ and  $\frac{\psi'_{\alpha_{4,2}-1,\alpha_{4,2}}(1,\lambda)}{\psi_{\alpha_{4,2}-1,\alpha_{4,2}}(1,\lambda)}$. Then, the potentials on $e_{8},e_{9},e_{10},e_{11}$ are determined by Lemma \ref{Th147}.   
    \begin{figure}[h]
    	\centering
    	\includegraphics[width=5.5cm,height=5cm]{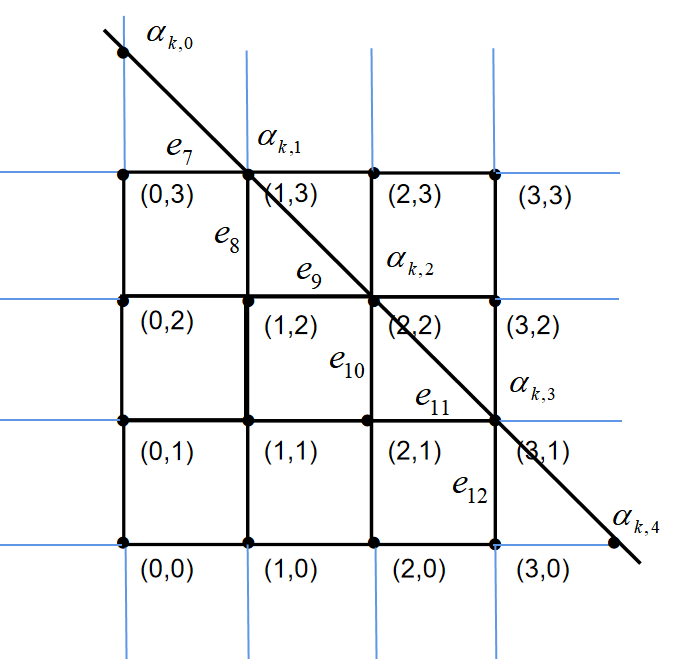}
    	\caption{N=3,k=4}
    	\label{7}
    \end{figure} 
    
    By step 6, rotate the whole system by the angle $\pi$ and the potentials on the rest of edges are determined.
    
    \begin{remark} \label{rem:compare}
    	Note that our reconstruction procedure is different from the one provided in \cite{And4, And5}. First, our algorithm depends on specific properties of the square lattice, while the procedure in \cite{And4, And5} was developed for the hexagonal lattice. Second, the reconstruction strategies differ. The authors of \cite{And4, And5} fix an arbitrary line $A_k$ and compute the solution $u$ of the boundary value problem generated by the reduced vertex Laplacian $(-\Delta_{V, \lambda} + Q_{V, \lambda})$ according to the analog of Lemma~\ref{SS2} for the hexagonal lattice. Then, they consider a zigzag line and find the ratios of the functions $\psi_{wv}(1,\lambda)$ for adjacent edges of that line. This helps to recover the potentials $q_e$. This method perfectly works for the discrete Schr\"odinger operators with constant coefficients (see \cite{And3, Iso2}). But the analogous treatment of the reduced vertex Laplacian causes difficulties, since the both $\Delta_{V, \lambda}$ and $Q_{V, \lambda}$ depend on some potentials $q_e$. Therefore, the solution $u$ cannot be found until the potentials are known. In our procedure, this problem does not arise, because we step-by-step reconstruct the potentials $q_e$ together with the reduced vertex Laplacian coefficients, which are needed at the next steps. A similar approach can be applied to the hexagonal lattice to improve the algorithm of \cite{And4, And5}.
    \end{remark}
    
	\section{Inverse scattering problem by the S-matrix} \label{sec:scattering}
	
	This section is concerned with the inverse scattering problem for the Schr\"odinger operator $H_E$ by the S-matrix. It has been shown in \cite{And4} that the edge D-N map uniquely specifies the S-matrix. Consequently, 
	the uniqueness theorem by the D-N map (Theorem~\ref{thm:edgeDN}) implies 
	the uniqueness of solution for the inverse scattering problem.
	In this section, we provide the definition of the S-matrix $S(\lambda)$, following the papers \cite{And2, And4, And5} and show that it uniquely determines the potential on the square lattice.
	
	In this section, it will be convenient for us to associate the vertex set $V$ of the square lattice with $\mathbb Z^2$. In other words, every vertex $v \in V$ is a pair of integers $n = (n_1, n_2)$.
	Define the vertex Laplacian 
	$$
	\left(\Delta_{V}f\right)(v)=\frac{1}{4}\sum_{w\sim v}f(w), \quad v\in V,
	$$
	which is self-adjoint on $L^{2}(V)$ equipped with the inner product
	$$
	(f,g)=4\sum_{n\in \mathbb{Z}^{2}}f(n)\cdot \overline{g(n)}.
	$$
	
	Put $\mathbb T^2 := \mathbb{R}^{2}\setminus (2\pi \mathbb{Z})^{2}$ and
	define the discrete Fourier transform $U_V \colon L^{2}(\mathbb{Z}^{2};\mathbb{C}^{2})\longrightarrow L^{2}(\mathbb{T}^{2};\mathbb{C}^{2})$ by
	\begin{equation}
		(U_{V}f)(x)=\frac{1}{\pi}\sum_{n\in \mathbb{Z}^{2}}e^{\mathrm{i}n\cdot x}f(n), \quad x=(x_{1},x_{2})\in \mathbb{T}^{2}. \label{U_V}
	\end{equation}
    The adjoint operator $U_V^* \colon L^{2}(\mathbb{T}^{2};\mathbb{C}^{2})\longrightarrow L^{2}(\mathbb{Z}^{2};\mathbb{C}^{2})$ is given by
    \begin{equation}
    	(U_{V}^{*})g(n)=\frac{1}{4\pi}\int_{\mathbb{T}^2}e^{-\mathrm{i} n\cdot x}g(x)dx, \quad n\in \mathbb{Z}^{2}.
    \end{equation}
	Then, on $L^{2}(\mathbb{T}^{2};\mathbb{C}^{2})$, $U_{V}(-\Delta_{V})U_{V}^{*}$ 
	is the operator of multiplication by
	\begin{equation}
		H_{0}(x)=-\frac{1}{2}(\cos x_{1}+\cos x_{2}).\label{H_0}
	\end{equation}

	For $q_e=0$, we denote $-\Delta_{V,\lambda}$ by $-\Delta_{V,\lambda}^{(0)}$. By Lemma 3.1 of \cite{And2}, we define the characteristic surface of $-\Delta_{V,\lambda}^{(0)}$ by
	$$
	M_{\lambda} = \{x \in \mathbb{T}^{2} \colon H_{0}(x)+\cos \sqrt{\lambda}=0\},
	$$
	which is smooth if $\cos \sqrt{\lambda}  \neq 0, \pm 1, \lambda \in \mathbb{R}$. We put
	$$
	T^{(0)}=\{\lambda \colon \cos \sqrt{\lambda} = 0, \pm 1 \},
	$$
	\begin{equation}
		T=T^{(0)}\cup \left( \cup_{e\in E}\sigma\left(H_{e}\right) \right),\label{T}
	\end{equation}
	where the operator $H_e$ was defined for equation \eqref{D-D} in Section \ref{sec:vertex}.
	
    The resolvent $r_{e}(\lambda) = \left(H_{e}-\lambda\right)^{-1}$
    is written as
    $$
    (r_{e}(\lambda)f)(z) = \int_{0}^{z}\frac{\phi_{e1}(z,\lambda)\phi_{e0}(t,\lambda)}{\phi_{e0}(1,\lambda)}f(t)dt+\int_{z}^{1}\frac{\phi_{e0}(z,\lambda)\phi_{e1}(t,\lambda)}{\phi_{e1}(0,\lambda)}f(t)dt.
    $$
    We put
    $$
    \Phi_{e0}(\lambda)f=\frac{d}{dz}(r_{e}(\lambda)f)\bigg|_{z=0}=\int_{0}^{1}\frac{\phi_{e1}(t,\lambda)}{\phi_{e1}(0,\lambda)}f(t)dt,
    $$
    $$
    \Phi_{e1}(\lambda)f=-\frac{d}{dz}(r_{e}(\lambda)f)\bigg|_{z=1}=\int_{0}^{1}\frac{\phi_{e0}(t,\lambda)}{\phi_{e0}(1,\lambda)}f(t)dt.
    $$
    
    Their adjoints acting from $\mathbb{C}$ to $L^{2}(0,1)$ are defined for $c \in \mathbb{C}$ by
    $$
    \Phi_{e1}(\lambda)^{*}c = c\frac{\phi_{e0}(z,\bar{\lambda})}{\phi_{e0}(1,\bar{\lambda})},\quad \Phi_{e0}(\lambda)^{*}c = c\frac{\phi_{e1}(z,\bar{\lambda})}{\phi_{e1}(0,\bar{\lambda})}.
    $$
    Define the operator  $T_{V}(\lambda):L^{2}_{loc}(E)\rightarrow L^{2}_{loc}(V)$ by
    \begin{equation}
    	(T_{V}(\lambda)f_{e})(v)=\frac{1}{4}\left(\sum_{e\in E_{v}(1)}\Phi_{e1}(\lambda)f_{e}+\sum_{e\in E_{v}(0)}\Phi_{e0}(\lambda)f_{e}\right),\quad v\in V.\label{T_V}
    \end{equation}
    The adjoint operator $T_{V}(\lambda)^{*}:L^{2}_{loc}(V)\rightarrow L^{2}_{loc}(E)$ has the form
    $$
    (T_{V}(\lambda)^{*}u)_{e}(z)=\Phi_{e1}(\lambda)^{*}u(e(1))+\Phi_{e0}(\lambda)^{*}u(e(0)).
    $$
    For $q_e=0$, we denote $T_{V}(\lambda)$ by $T_V^{(0)}(\lambda)$. Then,
    \begin{equation*} 
    \begin{split}
    (T_{V}^{(0)}(\lambda)f)(v)=&\frac{1}{4}\frac{\sqrt{\lambda}}{\sin \sqrt{\lambda}}\bigg(\sum_{e\in E_{v}(1)}\int_{0}^{1}\frac{\sin \sqrt{\lambda}z}{\sqrt{\lambda}}f(z)dz+\\
    &\sum_{e\in E_{v}(0)}\int_{0}^{1}\frac{\sin \sqrt{\lambda}(1-z)}{\sqrt{\lambda}}f(z)dz\bigg)
    \end{split}
    \end{equation*} 
    and 
    $$
    (T_{V}^{(0)}(\bar{\lambda})^{*}u)_{e}(z)=\frac{\sqrt{\lambda}}{\sin \sqrt{\lambda}}\left(\frac{\sin \sqrt{\lambda}z}{\sqrt{\lambda}}u(e(1))+\frac{\sin \sqrt{\lambda}(1-z)}{\sqrt{\lambda}}u(e(0))\right).
    $$
    
    Under the assumptions (Q-1), (Q-2), (Q-3), we define the Schr\"{o}dinger operator
    \begin{equation*}
    	H_{E} =\left\lbrace -\frac{d^{2}}{dz^{2}}+q_{e}(z)\colon e\in E\right\rbrace
    \end{equation*}
    with domain $ D(H_E)$ consisting of functions $u = \{ u_e\}_{e \in E}$ such that
    $u_e \in H^2(0,1)$ satisfy the Kirchhoff condition (K-1) and (K-2) in all the vertices $v \in V$ and $\sum_{e\in E} \|-\frac{d^2}{dz^2}u_e+q_e u_e\|_{L^2(0,1)}^2 < \infty$. Due to \cite{And4}, the operator $H_E$ is self-adjoint with the essential spectrum $\sigma_e(H_E) = [0, \infty)$. Furthermore, $\sigma_e(H_E) \setminus T$ is absolutely continuous, where $T \subset R$ is defined by \eqref{T}. Then, one can introduce the S-matrix $S(\lambda)$ for $\lambda \in (0, \infty) \setminus T$.
    
	Define the spaces $B^{*}(E)$ and $B_{0}^{*}(E)$ by
	\begin{align*}
	f \in B^{*}(E) \quad & \Leftrightarrow \quad \|f\|_{B^{*}(E)}^{2}=\sup_{R>1}\frac{1}{R}\sum_{|c(e)|<R}\|f_{e}\|_{L^{2}(0,1)}^{2}<\infty, \\
	f \in B_{0}^{*}(E) \quad & \Leftrightarrow \quad \lim_{R\rightarrow\infty}\frac{1}{R}\sum_{|c(e)|<R}\|f_{e}\|_{L^{2}(0,1)}^{2}=0,
	\end{align*}
	where $c(e)=\frac{1}{2}|e(0)+e(1)|$.
	For $f, g \in B^{*}(E)$, we use the notation $f \simeq g$ in the following sense:
	$$
	f \simeq g \Leftrightarrow f - g \in B_{0}^{*}(E).
	$$
	
	Now, we give the definition of the S-matrix $S(\lambda)$. The following lemma is a special case of Theorem 5.8 from \cite{And4}.
	
	\begin{lemma}[\cite{And4}]\label{S_lamb}
		Let $\lambda \in (0,\infty) \setminus T$. Then, for any incoming data $\phi^{in}\in L^{2}(M_{\lambda})$, there exist a unique solution $u\in B^{*}(E)$ of the equation
		$$
		( H_{E} -\lambda)u = 0,
		$$
		satisfying the Kirchoff conditions,
		and an outgoing data $\phi^{out}\in L^{2}(M_{\lambda})$
		satisfying
		\begin{align}
			u\simeq &-T_{V}^{(0)}(\lambda)^{*}U_{V}^{*}\frac{1}{\lambda(x)+\cos \sqrt{\lambda}+\mathrm{i}0\sigma(\lambda)}\phi^{in}\notag\\
			&+T_{V}^{(0)}(\lambda)^{*}U_{V}^{*}\frac{1}{\lambda(x)+\cos \sqrt{\lambda}-\mathrm{i}0\sigma(\lambda)}\phi^{out},
		\end{align}
		where $x=(x_{1},x_{2})$, $\lambda(x)=-\frac{1}{2}(\cos x_{1}+\cos x_{2})$ is the eigenvalue of $H_{0}(x)$ and $\sigma(\lambda) = 1$ if $\sin \sqrt{\lambda} > 0$, $\sigma(\lambda) = -1$ if $\sin \sqrt{\lambda} < 0$.
		
		The mapping
		$S(\lambda) \colon \phi^{in}\rightarrow \phi^{out}$
		is called the S-matrix.
	\end{lemma}
	
	The results of Section 6 in \cite{And4} imply the following lemma for the square lattice.
	
    \begin{lemma}[\cite{And4}] \label{lem:S-DN}
    	For the edge Schr\"{o}dinger operator on the square lattice, the S-matrix $S(\lambda)$ and the edge D-N map $\Lambda_E$ uniquely determine each other.
    \end{lemma}

    Therefore, Theorem~\ref{thm:edgeDN}, Lemma~\ref{lem:S-DN}, and the reconstruction procedure in Section~\ref{sec:construct} imply the following corollary.
	
\begin{corollary}\label{M1}
	Assume (Q-1), (Q-2) and (Q-3). Then, given any open interval $I \subset (0,\infty) \setminus T$, and the S-matrix $S(\lambda)$ for all $\lambda \in I$, one can uniquely reconstruct the potential $q_{e}(z)$ for all $e \in E$.
\end{corollary}

Indeed, under the assumptions (Q-1), (Q-2), and (Q-3), $S(\lambda)$ is meromorphic in the half-plane $\{ \lambda \colon \mbox{Re} \, \lambda > 0\}$ with possible branch points at $T$. Therefore, given $S(\lambda)$ for $\lambda \in I$, one can find $S(\lambda)$ for all $\lambda \in (0, \infty) \setminus T$ by analytic continuation.
Analogously, one can obtain the following result.

\begin{corollary}\label{M2}
	Assume (Q-1) and (Q-3), and given a real $q_{0}(z) \in L^{2}(0,1)$ satisfying $q_{0}(z) = q_{0}(1 - z)$ and $q_{e}(z) = q_{0}(z)$ on $(0, 1)$ except for a finite number of edges $e \in E$. Given an open interval $I \subset \sigma_e(H_E) \setminus T$ and the S-matrix $S(\lambda)$ for all $\lambda \in I$, one can uniquely reconstruct the potential $q_{e}(z)$ for all edges $e \in E$.
\end{corollary} 

The reconstruction procedure for Corollary~\ref{M2} requires no essential changes. Instead of $\frac{\sin \sqrt{\lambda}z}{\sqrt{\lambda}}$ and
$\frac{\sin \sqrt{\lambda}(1-z)}{\sqrt{\lambda}}$, we have only to use the corresponding solutions to the Schr\"{o}dinger
equation $(-\frac{d^{2}}{dz^{2}}+q_{0}(z)-\lambda)\varphi=0$.

Corollaries~\ref{M1} and~\ref{M2} are analogous to the results of \cite{And4, And5} for the hexagonal lattice.

\end{document}